\documentclass{article}

\usepackage[superscript,biblabel]{cite}
\usepackage{authblk}
\usepackage{amsmath, amsfonts, amssymb, mathtools, amsthm}
\usepackage{booktabs}
\usepackage{multirow}
\usepackage{hyperref}
\usepackage{bbm}

\newcommand{\Var}{\text{Var}}
\newcommand{\Cov}{\text{Cov}}

\newtheorem{lemma}{Lemma}

\DeclareSymbolFont{bbold}{U}{bbold}{m}{n}
\DeclareSymbolFontAlphabet{\mathbbold}{bbold}

\title{On Variance Estimation for the One-Sample Log-Rank Test}

\author[1,*]{Moritz Fabian Danzer}
\author[1]{Andreas Faldum}
\author[1]{Rene Schmidt}

\affil[1]{Institute of Biostatistics and Clinical Research, University of Muenster, Germany}
\affil[*]{Corresponding author, moritzfabian.danzer@ukmuenster.de}

\date{\today}

\begin{document}
	
\maketitle

\begin{abstract}
	Time-to-event endpoints show an increasing popularity in phase II cancer trials. The standard statistical tool for such one-armed survival trials is the one-sample log-rank test. Its distributional properties are commonly derived in the large sample limit. It is however known from the literature, that the asymptotical approximations suffer when sample size is small. There have already been several attempts to address this problem. While some approaches do not allow easy power and sample size calculations, others lack a clear theoretical motivation and require further considerations. The problem itself can partly be attributed to the dependence of the compensated counting process and its variance estimator. For this purpose, we suggest a variance estimator which is uncorrelated to the compensated counting process. Moreover, this and other present approaches to variance estimation are covered as special cases by our general framework. For practical application, we provide sample size and power calculations for any approach fitting into this framework. Finally, we use simulations and real world data to study the empirical type I error and power performance of our methodology as compared to standard approaches.
\end{abstract}

\textbf{Keywords:} one-sample log-rank test, phase II trial, reference distribution, single-arm study, survival analysis

\section{Introduction}
As recently shown, \cite{Ivanova:2016} time-to-event endpoints, such as overall survival (OS) and progression-free survival (PFS), experience increased popularity in Phase II oncology trials. This may be due to the changes in clinical development of oncology treatments as new treatment types relying on different mechanisms come up. Additionally, a majority of these Phase II trials is single-armed \cite{Ivanova:2016}. However, the one-sample log-rank test, which is widely used in these situations, is known to be quite conservative \cite{Wu:2015}. This is the case especially for small sample sizes. There have been some efforts to solve this problem. Unfortunately, some of these methods are quite complicated as they require estimations of higher moments and are not suitable for sample size calculation as the distribution of the test statistic under alternative hypotheses is unknown, \cite{Tu:1996, Sun:2011} while others lack a clear theoretical motivation and can be anti-conservative in some cases \cite{Wu:2016, Wu:2014}. Nevertheless, the approach by Wu \cite{Wu:2014} sheds a light on the previously neglected possibility to include the counting process into the variance estimation.\\
The conservativeness of the classical approach as well as the tendency towards anti-conservativeness of the new approach \cite{Wu:2014} in many scenarios is due to the skewness of the one-sample log-rank statistic when sample size is small. A central cause for this skewness is the dependence between the unstandardized one-sample log-rank statistic and those variance estimators commonly used for standardization. Because while Basu's theorem \cite{Basu:1955} guarantees independence of mean and variance estimator for normally distributed data, this does not generally apply to the survival setting. However, the theoretical result on which the asymptotical correctness of the test is based, \cite{Andersen:93} leaves us with one degree of freedom concerning the choice of the variance estimator. We develop a framework in which this property is taken advantage of. Exploiting this degree of freedom, we find a variance estimator which is again uncorrelated to the compensated counting process under a fixed hypothesis of the testing problem. In addition, the classical one-sample log-rank test and already existing improvements \cite{Wu:2016, Wu:2014} can be embedded as special cases into it.\\
Especially recently developed methodology concerning adaptive \cite{Schmidt:18}, multi-stage \cite{Shan:2019, Belin:2017, Kwak:2014} and multivariate extensions \cite{Danzer:2021} of the one-sample log-rank test require and benefit from a proper behaviour of the underlying simple one-sample log-rank test as they rely on the concordance of the distribution of the test statistics with the normal distribution for even smaller quantiles.\\
Our approach is neither computationally intense nor requires any information beyond the input data requested by standard software as PASS \cite{PASS:2018} or nQuery \cite{nQuery:2017} for designing classical one-sample log-rank tests. Hence, existing tools for planning and execution of the one-sample log-rank test can easily be extended to incorporate any approach fitting into our framework.\\
The paper is organized as follows. We construct our framework in Section \ref{section:definitions} after settling basic notation. Afterwards, we present power and sample size calculations. In Section \ref{section:optimal_weight} we propose the uncorrelated variance estimator. Existing approaches and their type I error rate and power performance for small sample sizes are compared by simulations in Section \ref{section:simulations}. We finally conduct a case study in Section \ref{section:real_world_example} in order to give an example for an application of the procedure. We conclude with a discussion in which we give some advice for planning and execution of one-sample log-rank tests.\\

\section{Testing problem and significance tests}\label{section:definitions}
In a single-arm study with $n$ patients, let $Y_i$ denote the calendar time of entry of the $i$-th patient and $T_i$ its time from entry until event. In what follows, it will be important to distinguish between censoring that occurs at a fixed calendar date $t$ of analysis, given by $(t-Y_i)_+\coloneqq \max(t-Y_i,0)$, and random dropouts, denoted by $C_i$. At calendar time $t$, only $T_i \wedge C_i \wedge (t-Y_i)_+ \coloneqq \min(T_i, C_i, (t-Y_i)_+)$ will be observable for any $i\in \{1,\dots,n\}$. We assume $T_i,\;C_i$ and $Y_i$ to be mutually independent for any $i$ and the tuples $\{(T_i, C_i, Y_i), i\in \{1,\dots,n\}\}$ to be independent and identically distributed with the same distribution as the tuple $(T,C,Y)$, say. 
Let $F_T$ denote the true unknown distribution function of the survival random variable $T$. Analogously, $F_Y$ and $F_C$ are the true unknwon distribution functions of $Y$ respectively (resp.) $C$. We also fix a reference distribution function $F_{T,0}$ that describes the expected distribution of $T$ under standard of care while $F_{T,1}$, which is the planning alternative underlying sample size calculations, represents the expected distribution of $T$ under a new treatment. $F_{T,0}$ is commonly chosen in the light of historical data under a standard treatment.\\ 
For any distribution function of a non-negative random variable $X$ with distribution function $F_X$ we can compute their density $f_X$, survival function $S_X$, hazard $\lambda_X$ and cumulative hazard $\Lambda_X$ via the following well-known identities:
\begin{align}\label{eq:general_dependencies}
\begin{split}
&\frac{d}{dt}F_X(t)=f_X(t),\;F_X(t)=1-S_X(t),\; \lambda_X(t)=\frac{f_X(t)}{S_X(t)}\text{ and}\\
&\Lambda_X(t)=\int_0^t \lambda_X(s)\;ds=-\text{log}(1-F_X(t)).
\end{split}
\end{align}
For example, $\Lambda_{T,0}(t)=\log(1-F_{T,0}(t))$ denotes the reference cumulative hazard function. In the sequel, we consider testing the two-sided hypothesis
\begin{equation}
H_0 \colon \Lambda_T = \Lambda_{T,0}
\end{equation}
which is equal to the intersection of the two one-sided hypotheses
\begin{equation}
H_{0,1}\colon \Lambda_{T} \geq \Lambda_{T,0} \text{ and } H_{0,2}\colon \Lambda_{T} \leq \Lambda_{T,0}.
\end{equation}
The number of events observed at calendar time $t$ given $Y_i=y_i$ for any $i$ is
\begin{equation}\label{eq:counting_process}
N(t)\coloneqq\sum_{i=1}^n \underbrace{\mathbbold{1}_{\{T_i\leq(t-y_i)_+ \wedge C_i\}}}_{\eqqcolon N_i(t)}.
\end{equation}
The number of expected events under the null hypothesis $H_0$ given $Y_i=y_i$ for any $i$ is
\begin{equation}\label{eq:compensator_under_H0}
A_0(t)\coloneqq \sum_{i=1}^n \underbrace{\Lambda_{T,0}(T_i \wedge C_i \wedge (t-y_i)_+)}_{\eqqcolon A_{0,i}(t)}.
\end{equation}
We define additionally
\begin{equation}\label{eq:martingale_under_H0}
M_0(t)\coloneqq n^{-\frac{1}{2}} \left(N(t)-A_0(t)\right) \quad \forall t\geq 0.
\end{equation}
It is known from the literature that $(M_0(t))_{t \geq 0}$ is a martingale under the null hypothesis with respect to (w.r.t.) the filtration generated by the survival processes of the $n$ patients.
It follows from Theorem II.5.1 of \textit{Andersen et al.}\cite{Andersen:93} that this martingale converges in distribution against a continuous Gaussian martingale as the number of patients $n$ converges to infinity. In particular,
\begin{equation}
M_0(t) \overset{\mathcal{D}}{\underset{n \to \infty}{\to}} \mathcal{N}(0,V(t))
\end{equation}
for any $t\geq 0$ with $V(t)=P_{H_0}[T \leq C \wedge (t-Y)_+]$. The same theorem also leaves us with two possible choices to estimate the non-decreasing covariance function $V(t)$ of the limiting process under the null hypothesis because
\begin{equation}
\frac{1}{n} N(t) \overset{P}{\underset{n \to \infty}{\rightarrow}} V(t)
\end{equation}
as well as
\begin{equation}
\frac{1}{n} A_0(t) \overset{P}{\underset{n \to \infty}{\rightarrow}} V(t).
\end{equation}
Consequently, by a linear combination with a weight function $w\colon [0,\infty) \to [0,1]$, we also have
\begin{equation}\label{eq:generalcorrectvarianceestimator}
\hat{\sigma}_w^2(t) \coloneqq \frac{1}{n} \left(w(t) \cdot N(t) + (1-w(t)) \cdot A_0(t) \right) \overset{P}{\underset{n \to \infty}{\rightarrow}} V(t)
\end{equation} 
pointwise for any $t \in [0,\infty)$. By Slutsky's theorem, we conclude that 
\begin{equation}\label{eq:convergence_test_statistic}
\frac{M_0(t)}{\hat{\sigma}_w(t)}\overset{\mathcal{D}}{\rightarrow}Z
\end{equation}
pointwise in $t\geq 0$ under $H_0$, where $Z\sim \mathcal{N}(0,1)$. Thus, we obtain an asymptotically correct two-sided test with type I error level $\alpha$ of the hypothesis
$H_0 \colon \Lambda_T = \Lambda_{T,0}$ by rejecting it if
\begin{equation}\label{eq:two_sided_test}
\left| \frac{M_0(t)}{\hat{\sigma}_w(t)} \right| \geq \Phi^{-1}\left(1-\frac{\alpha}{2}\right)
\end{equation} 
where $M_0$ is defined as in \eqref{eq:martingale_under_H0}. Analogously one obtains a one-sided test with type I error level $\alpha/2$ of the null hypothesis $H_{0,1} \colon \Lambda_T \geq \Lambda_{T,0}$ by rejecting it if
\begin{equation}
\frac{M_0(t)}{\hat{\sigma}_w(t)} \leq \Phi^{-1}\left(\frac{\alpha}{2}\right).
\end{equation} 
It should be emphasised here that the functional form of $w(t)$ must be determined in advance and must not be changed in the course of the study. Under this condition, distributional convergence is guaranteed pointwise in $t$ according to formula \eqref{eq:convergence_test_statistic} independent from the specific choice of $w(t)$. By suitable choice of $w(t)$, small sample properties can be improved (see Section \ref{section:optimal_weight}).
Of course, it is easy to embed the original choice and Wu's suggestion \cite{Wu:2014} into this framework as we have $w_{orig}\equiv 0$ resp. $w_{Wu}\equiv 0.5$ in these cases which lead to
\begin{align}
&\hat{\sigma}_{w_{orig}}^2(t)\coloneqq \frac{1}{n} A(t) \qquad \qquad \text{resp.}\label{eq:orig_approach_new_framework}\\
&\hat{\sigma}_{w_{Wu}}^2(t)\coloneqq \frac{1}{2n}A(t) + \frac{1}{2n}N(t)\label{eq:Wus_approach_new_framework}.
\end{align}
Further choices are discussed in Section \ref{section:optimal_weight} and compared via simulation in Section \ref{section:simulations}.
\section{Power and sample size calculation}\label{section:sample_size}
Power and sample size will be calculated under some planning alternative $H_1\colon \Lambda_T= \Lambda_{T,1}$, say, and additional planning assumptions concerning the distribution of $Y$ and $C$. Corresponding calculations for our general approach can be adopted from previous power and sample size calculations \cite{Wu:2015}. With a fixed analysis date $t>0$ and the abbreviation
\begin{equation}\label{eq:combined_censoring}
S_{U}(s)\coloneqq S_C(s)\cdot F_Y((t-s)_+)
\end{equation}
we have:
\begin{equation}
\begin{aligned}
\mathbb{E}_{H_1}[N_i(t)]&=\int_0^t S_{U}(s)f_{T,1}(s)ds&&\eqqcolon v_1(t),\\
\mathbb{E}_{H_1}[A_{0,i}(t)]&=\int_0^t S_{U}(s)S_{T,1}(s)\lambda_{T,0}(s)ds&&\eqqcolon v_0(t),\\
\mathbb{E}_{H_1}[N_i(t)\cdot A_{0,i}(t)]&=\int_0^t S_{U}(s)f_{T,1}(s)\Lambda_{T,0}(s)ds&&\eqqcolon v_{01}(t),\\
\mathbb{E}_{H_1}[A_{0,i}(t)^2]&=2\int_0^t S_{U}(s)S_{T,1}(s)\Lambda_{T,0}(s)\lambda_{T,0}(s)ds&&\eqqcolon 2v_{00}(t).
\end{aligned}
\end{equation}
for any $i$. Note that $S_U$ is the survival function of the random variable $U\coloneqq C \wedge (t-Y)_+$. The expectation and variance of $M_0(t)$ under the alternative hypothesis amount to
\begin{equation}
\begin{aligned}
\mathbb{E}_{H_1}[M_0(t)]&=\sqrt{n}(v_1(t) - v_0(t))&&\eqqcolon \sqrt{n}\omega(t) \quad\text{resp.}\\
\Var_{H_1}(M_0(t))&=v_1(t)  - v_1^2(t) + 2v_{00}(t) - v_0^2(t) - 2v_{01}(t) + 2 v_0(t) v_1(t)&&\eqqcolon \sigma^2.
\end{aligned}
\end{equation}
For our variance estimator with the prefixed weight function $w$, we have
\begin{equation}\label{eq:variance_estimator_under_alternative}
\hat{\sigma}_w^2(t) \overset{P}{\underset{n \to \infty}{\rightarrow}} w(t) v_1(t) + (1-w(t)) v_0(t) \eqqcolon \bar{\sigma}_w^2(t). 
\end{equation}
under the alternative hypothesis $H_1$. Now, it follows from Slutsky's theorem that
\begin{equation}
\frac{M_0(t)}{\hat{\sigma}_w(t)} - \frac{\sqrt{n}\omega(t)}{\bar{\sigma}_w(t)}\underset{n\to\infty}{\overset{\mathcal{D}}{\rightarrow}} \mathcal{N}\left(0,\frac{\sigma^2(t)}{\bar{\sigma}_w^2(t)}\right)
\end{equation}
under $H_1$. Given these quantities, we obtain a sample size of
\begin{equation}\label{eq:sample_size}
n=\frac{(\bar{\sigma}_w(t)\cdot z_{1-\frac{\alpha}{2}} + \sigma(t) \cdot z_{1-\beta})^2}{\omega^2}
\end{equation}
for a two-sided test with nominal type I error level $\alpha$ and power $1-\beta$ where $z_q$ denotes the $q$-quantile of a standard normal distribution and $t$ is the prefixed calendar time of analysis.\\
We already recognize, that a decrease of $\bar{\sigma}(t)$ in \eqref{eq:sample_size} leads to a decrease of the sample size required for fixed type I \& II errors. As typically $\lambda_{T,1}(s) \leq \lambda_{T,0}(s)$ for all $s \in \mathbb{R}_+$, also notice that we have $v_1<v_0$. In terms of \eqref{eq:variance_estimator_under_alternative} it thus seems to be advisable to choose $w\equiv 1$. For the reasons explained in Section \ref{section:optimal_weight}, however, this choice leads to anti-conservative testing procedures for small sample sizes.\\
Please also note that \eqref{eq:sample_size} is only an implicit formula if the accrual rate $r=n/a$ is the quantity which is given externally, where $a$ denotes the length of the recruitment period. Its solution also requires assumptions on the distribution of the accrual and dropout random variables $Y$ and $C$. Common planning assumptions are:
\begin{enumerate}
	\item uniform recruitment during a period of length $a$ with constant rate $r$, i.e. $n=r\cdot a$ and $f_Y(s)=1/a\cdot \mathbbm{1}_{[0,a]} (s)$, followed by an observation period of length $f$ up until the calendar time $t=a+f$ of analysis,
	\item no additional dropouts, i.e. $C=+\infty$ or $S_C \equiv 1$.
\end{enumerate}
With these specifications, formula \eqref{eq:sample_size} can be solved for the only remaining parameter $a$ using standard numerical methods. The required sample size is then $n=r \cdot a$.
\section{Choice of the weight function}\label{section:optimal_weight}
Theoretical considerations in Section \ref{section:definitions} showed, that the test procedure \eqref{eq:two_sided_test} is an asymptotically exact level $\alpha$ test for any weight function $w$. This will also be confirmed by the simulations in Section \ref{section:simulations_variable_weight}. However, the small sample performance strongly depends on the choice of $w$. Whereas some choices of $w$ imply a conservative test, others lead to anti-conservatism when sample size is small. We now want to give advice on how to choose $w$ in order to improve small sample properties of the test with an actual type I error rate close to the nominal one.
\subsection{Uncorrelated variance estimators}
This section deals with one of the two causes for the skewness of the distribution of the test statistic $M_0(t)/ \hat{\sigma}_w(t)$ for small sample sizes: the correlation of $M_0$ and $\hat{\sigma}_w^2(t)$ (see \eqref{eq:martingale_under_H0} resp. \eqref{eq:generalcorrectvarianceestimator}). This problem can be approached in an analytical way by a suitable choice of $w$.\\
It is well-known that such a correlation causes a skewness of the ratio of $M_0$ and $\hat{\sigma}_w(t)$. This problem has been observed for several cases in which numerator and denominator themselves are symmetrically distributed while their correlation causes a skewness of the ratio \cite{Hinkley:1969, Nadarajah:2006, Ly:2019}. In particular, positive correlation causes a left-skew and negative correlation causes a right-skew of the emerging distribution. This leads to an increase of the probability mass on the left tail and a decrease of the probability mass on the right tail in the former case while in the latter case it is just the other way round. Because of this, it is important to not just examine empirical type I error levels of two-sided tests in simulation studies, but also to consider how this empirical type I error level is distributed among the two underlying one-sided tests.\\
As $\Cov_{H_j}(N(t)-A_0(t),A_0(t)) < \Cov_{H_j}(N(t)-A_0(t),N(t))$ for any $j \in \{0,1\}$ and $t\in \mathbb{R}_+$ in any trial where there is almost surely some patient with a positive length of stay, there is a $w_j(t) \in \mathbb{R}$ such that (s.t.)
\begin{equation}\label{eq:defining_equation}
\Cov_{H_j}(N(t) - A_0(t),w_j(t)\cdot N(t) + (1-w_j(t)) \cdot A_0(t))=0\\
\end{equation}
for $j \in \{0,1\}$. As $\text{Cov}_{H_j}(A_0(t),N(t)) < 0$ (see Appendix, Lemma 1), we have $w_j(t)\in [0,1]$ for any $t\geq 0$. Using $w_j(t)$ in \eqref{eq:generalcorrectvarianceestimator}, we obtain a consistent variance estimator $\hat{\sigma}_{w_j}^2(t)$ which by definition is uncorrelated with the numerator of our test statistics under $H_j$. As we will see in the simulations, the choice $\hat{\sigma}_{Wu}^2 (t) = A(t)/2n + N(t)/2n$ is a good first choice w.r.t. a decrease of the correlation of numerator and denominator, which may be improved using the weight function $w_j$ defined by \eqref{eq:defining_equation}.\\
Obviously, \eqref{eq:defining_equation} is equivalent to
\begin{equation}
w_j(t)=-\frac{\Cov_{H_j}(N(t)-A_0(t),A_0(t))}{\Var_{H_j}(N(t)-A_0(t))}
\end{equation}
and
\begin{equation}
1-w_j(t)=\frac{\Cov_{H_j}(N(t)-A_0(t),N(t))}{\Var_{H_j}(N(t)-A_0(t))}.
\end{equation}
For $j=0$ we obtain
\begin{equation}\label{eq:ordinary_weight}
w_0(t)=\frac{\int_0^{t} S_{U}(s)f_{T,0}(s)\Lambda_{T,0}(s)ds}{\int_0^{t} S_{U}(s)f_{T,0}(s)ds}=1-\frac{\int_0^t S_{T,0}(s)\Lambda_{T,0}(s)dS_{U}(s)}{\int_0^t F_{T,0}(s)dS_{U}(s)}.
\end{equation}
Recall that by \eqref{eq:combined_censoring}, $S_U$ depends on the prefixed calendar date of analysis $t$, too. With the abbreviations introduced in Section \ref{section:sample_size} we get for $j=1$, i.e. the scenario of the planning hypothesis
\begin{equation}
w_1(t)=\frac{2v_{00} - v_0^2 - v_{01} + v_0v_1}{v_1 - v_1^2 + 2v_{00} - v_0^2 - 2v_{01} +2v_0v_1}.
\end{equation}
In what follows, we will limit ourselves to the consideration of $w_0$ as it is our primary aim to control the type I error rate of the trial.

\subsection{Practical implementation}
In the application of \eqref{eq:ordinary_weight} we are confronted with the problem that $w_0$ depends (via $S_U$) on the true but unknown distribution of the recruitment and dropout random variables $Y$ and $C$. The weight $w_0$ realizing exact zero correlation of $M_0(t)$ and $\hat{\sigma}^2_{w_0}(t)$ is thus unknown, too. In practice we thus have to use the weight function $\tilde{w}_0$ that is calculated according to \eqref{eq:ordinary_weight} after replacing the true distributions of $Y$ and $C$ by best possible planing assumptions about them.
For testing $H_0$, we thus propose to use the test statistic $M_0(t)/\hat{\sigma}_{\tilde{w}_0}(t)$ as best possible approximation to the actually desired, unknown $M_0(t)/\hat{\sigma}_{w_0}(t)$. On first sight, using $M_0(t)/\hat{\sigma}_{\tilde{w}_0}(t)$ as test statistic appears unusual because it formally depends on distributional assumptions beyond the null hypothesis itself. In this context, however, it is crucial to note that $\tilde{w}_0$ is a valid weight function despite of being chosen based on assumptions about $Y$ and $C$. Consequently, our general results from Section \ref{section:definitions} assure that the statistic $M_0(t)/\hat{\sigma}_{\tilde{w}_0}(t)$ realizes an exact level $\alpha$ test of $H_0$ as sample size increases. A misspecification of the distributional assumptions about $Y$ and $C$ might only lead to a small non-zero correlation between $M_0(t)$ and $\hat{\sigma}^2_{\tilde{w}_0}$ because then $w_0 \neq \tilde{w}_0$. Asymptotically, the type I error rate control is not impaired by such misspecifications.\\
For optimal small sample size performance, however, it should be on the investigator's mind to specify all planning assumptions as well as possible so that $\tilde{w}_0$ is close to $w_0$. Furthermore, it should be borne in mind that these assumptions are relied on in the sample size calculation. Reassuringly, our results below suggest that even major misspecifications of the distributions of $Y$ and $C$ do not have a major impact on the weight suggested in \eqref{eq:ordinary_weight}.\\
Additionally, the resulting testing procedure will be at most as conservative as the standard one-sample log-rank test (i.e. $w\equiv 0$, see \eqref{eq:orig_approach_new_framework}) and at most as anti-conservative as the test which is based on choosing $w\equiv 1$ \cite{Chu:2020}.\\ 

\subsection{Consequences of misspecifications}\label{section:misspecifications}
To illustrate the impact of a misspecification of the distribution of $Y$ and $C$ on the weight $w_0(t)$, we consider the following numerical example: Let the null hypothesis be given by \begin{equation}
H_0\colon \Lambda_T(s)=-\log(1/2)s
\end{equation}
i.e. $T$ is exponentially distributed with a median survival time of two years. In a trial with an accrual duration of one year and a follow-up duration of one year, we compute the weight under the assumption of uniform distribution of $Y$ on $[0,1]$ and a random dropout rate of $-\log(9/10)$ which leads to a dropout of 10\% of the patients each year. In this case, we obtain $w_0(2)=0.4215$ according to \eqref{eq:ordinary_weight}. In contrast, yearly dropout rates of 0 or 30\% lead to the weights $w_0(2)=0.4359$ resp. $w_0(2)=0.3891$. To illustrate a misspecification of the distribution of $Y$, we consider the case that its distribution function on $[0,1]$ is not given by $F_Y(s)=s$, but by $F_Y(s)=s^{\theta}$ for some $\theta>0$. In the case of very early recruitment ($\theta=1/2$), we obtain $w_0(2)=0.4556$ and in the case of very late recruitment ($\theta=2$), we obtain $w_0(2)=0.3844$. If the recruitment speed is incorrectly estimated and the length of the recruitment period amounts to 0.5 or 1.5 years, the weights $w_0(2)=0.4699$ resp. $w_0(2)=0.3770$ are obtained. Since the weight can be selected between 0 and 1, these changes are not only numerically small, but also have only very small effects on small sample properties as we will see in Section \ref{section:simulations_variable_weight}.

\section{Simulation study}\label{section:simulations}
In this section we want to shed a light on multiple aspects of the topics considered so far. In a first simulation we want to illustrate the influence of the choice of the weight on the small sample properties of the one-sample log-rank test and its asymptotics. In Section \ref{section:emp_typeI_and_power} we compare log-rank tests that result from different choices of the variance estimator in terms of sample sizes and empirical errors for a wide range of scenarios.\\ 
As already explained, a skewed distribution of the test statistic implies a lack of concordance with the normal distribution on both tails. Nevertheless, it is possible that the deviations from the nominal level on both tails cancel each other out and it may seem that the empirical error level is close to the nominal level for the two-sided test, although the test may be misbehaving at each tail. An example for such a behaviour can be found in Section \ref{section:real_world_example}. Therefore we primarily focus on one-sided tests of $H_{0,1}$, whose rejection would result in the acceptancce of the superiority of the new treatment. As we naturally carry out two-sided tests with a nominal type I error level of $2\alpha=5\%$, we consider the left-sided tests with a nominal level of $\alpha = 2.5\%$ in what follows. All simulations were performed using R, version 4.0.2 \cite{R}.

\subsection{Influence of the choice of the weight}\label{section:simulations_variable_weight}
In this simulation study, we consider the empirical type I error levels for a variety of sample sizes and weights. We intend to confirm that the asymptotics apply independently of the choice of the weight, but that the weight influences the actual type I error level for small sample sizes. Hence, we consider the properties of the test statistics introduced above under the null hypothesis.\\ 
The accrual duration $a$ and the follow-up duration $f$ of the trial are both set ato $1$ year, i.e. the final analysis takes place after $2$ years. We assumed uniform accrual and a random dropout rate of $10\%$ per year. We considered seven different sample sizes (25, 50, 100, 250, 500, 1000, 5000) and computed the values of the test statistic for 101 different weights, ranging from 0 to 1 in steps of 0.01. The survival variable $T$ is exponentially distributed with parameter $\lambda_{0,T}$. We considered four different parameter values for $\lambda_{0,T}$. They are chosen in such a way that the expected event rates of these trials amount to 20\%, 40\%, 60\% and 80\% per scenario, respectively.\\
We executed 100,000 simulation runs per scenario. In this case, the breadth of a 95\%-confidence interval for an underlying true value of 0.025 is about 0.002. The results are displayed in Figure \ref{fig:asymptotics}.

\begin{figure}[h]
	\centering
	\includegraphics[width=\textwidth]{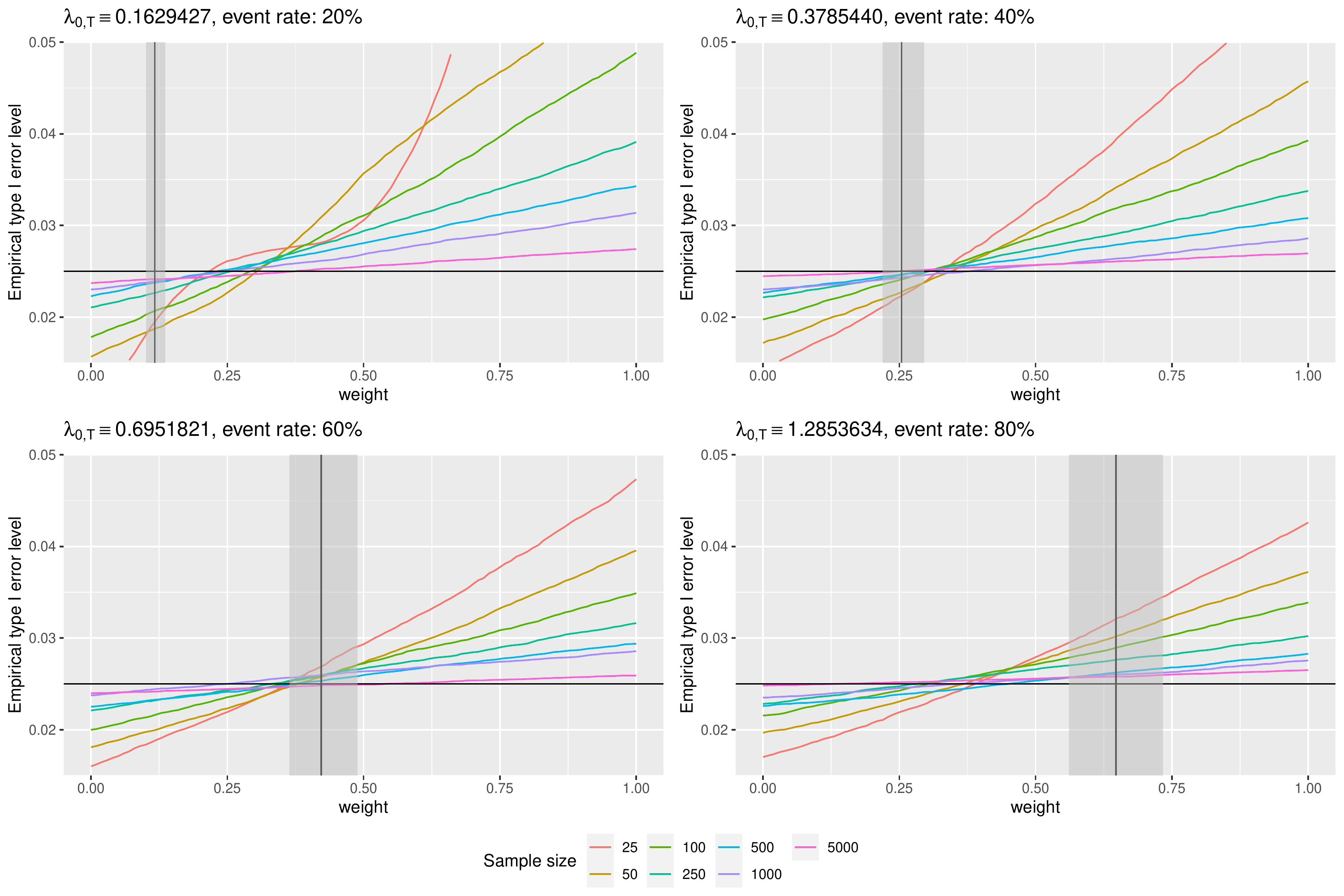}
	\caption{Empirical type I error rates in four different scenarios with event rates of 20\%, 40\%, 60\% and 80\% respectively and weights varying between 0 and 1. The horizontal line marks the nominal level of 0.025, the vertical line marks the weight for which $M_0(2)$ and its variance estimator are uncorrelated. The shaded area accounts for misspecifications of the accrual and censoring mechanism.}
	\label{fig:asymptotics}
\end{figure}

As expected, the convergence to the nominal type I error level for increasing sample size can be observed for any weight and in any scenario. In scenarios with low or medium event rates, the uncorrelated variance estimator from Section \ref{section:optimal_weight} works very well, whereas for high event rates, the approach by Wu \cite{Wu:2014} performs more favorable. Actually, a combination of Wu's and our variance estimator promises a nearly optimal performance concerning the type I error rate as this may prevent from anti-conservativeness in cases with high event rates. Such a combination corresponds to choosing the weight according to 
\begin{equation}\label{eq:combination_Wu_uncorrelated}
w(t)=\min\left(-\frac{\Cov_{H_0}(N(t)-A_0(t),A_0(t))}{\Var_{H_0}(N(t)-A_0(t))}, 0.5 \right).
\end{equation}\\
The shaded areas in the plots of Figure \ref{fig:asymptotics} describe a range of values of the weights which are calculated via \eqref{eq:ordinary_weight} under misspecifications of accrual and censoring mechanisms. Values in this area are obtained if the yearly dropout rate is expected to be between 0 and 20\% and the accrual period lasts between 0.5 and 1.5 years. These are quite severe misspecifications but nevertheless, they only lead to minor deviations with regard to the weight and the empirical type I error rate. In the large sample limit, type I error rate control is guaranteed irrespective of the chosen weight.\\

\subsection{Empirical type I error level and power performance}\label{section:emp_typeI_and_power}
We conduct simulations in order to compare the performance of the different approaches for variance estimation concerning sample size, empirical type I and type II error of the resulting one-sample log-rank tests. To ensure comparability to already existing literature on this topic, we considered scenarios inspired by previous simulation studies \cite{Wu:2015, Wu:2014}.\\
Hence, the survival distribution under the null hyothesis is taken as a Weibull distribution with distribution function $S_0(t)=\exp(-\log(2)(t/m_0)^\kappa)$ and cumulative hazard function $\Lambda_0(t)=\log(2)(t/m_0)^{\kappa}$ where the shape parameter $\kappa$ and the median survival time $m_0$ are given. We assume that the survival time under the alternative is also given by a Weibull distribution with the same shape parameter $\kappa$ and a different median survival time $m_1$, which is determined by the hazard ratio $\delta$ through $\delta=(m_1/m_0)^{\kappa}$. The accrual mechanism is given by an uniform accrual during a period of length $a=1$ with a follow-up period of length $f=3$ corresponding to a final analysis at calendar time $t=4$. We do not consider additional random drop-outs here.\\
We will investigate shape parameters $\kappa \in \{0.1, 0.25, 0.5, 1, 2, 5\}$, but in deviation from other publications \cite{Wu:2015, Wu:2014}, we will not only consider the fixed mean survival time $m_0=1$, but values of $m_0 \in \{1,2,4\}$. On the other hand, we restrict the possible values of $\delta$ to one small ($\delta = 1.2$), one medium ($\delta = 1.5$) and one large ($\delta= 2$) hazard ratio.\\
\begin{table}
	\centering
	\caption{Expected event rates under the null hypothesis and weights for the uncorrelated variance estimator for different combinations of shape parameter $\kappa$ and median survival time $m_0$}
	\begin{tabular}{c c c c c c c}
		\toprule
		&\multicolumn{6}{c}{$\kappa$}\\
		\cmidrule{2-7}
		$m_0$&0.1&0.25&0.5&1&2&5\\
		\midrule
		\multirow{2}{*}{1}&52.98\%&57.58\%&65.31\%&78.96\%&91.52\%&97.18\%\\
		&0.3307&0.3706&0.4481&0.6280&0.8626&0.9599\\
		\multirow{2}{*}{2}&50.55\%&51.40\%&52.98\%&56.04\%&61.85\%&67.35\%\\
		&0.3114&0.3199&0.3383&0.3897&0.5324&0.8062\\
		\multirow{2}{*}{4}&48.16\%&45.50\%&41.39\%&34.43\%&24.85\%&12.89\%\\
		&0.2931&0.2750&0.2504&0.2175&0.1873&0.1664\\
		\bottomrule
	\end{tabular}
	\label{table:expected_events_and_weights}
\end{table}
The reason for extending the range of median survival times under the null hypothesis can be seen in Table \ref{table:expected_events_and_weights}. As the first row indicates, $m_0=1$ leads to trials with high event rates: In each of these scenarios event rates of more than 50\% are expected by time $t=4$ and two of the six scenarios lead to event rates of more than 90\% by time $t=4$. By including larger median event times ($m_0 \in \{2,4\}$), we want to broaden the range of expected event rates under the null hypothesis by time of analysis. This way we can clearly distinguish which approach is most useful in the respective setting.\\
Again we focused on the empirical errors of the one-sided tests with nominal level of $2.5\%$. The sample size was planned with a power of $80\%$. The results can be found in Table \ref{table:samplesizes_emp_erorrs}. We conducted 100 000 simulation runs for each scenario.\\
\begin{table}
	\caption{Sample sizes and empirical type I and type II errors for the four different approaches of variance estimation in all considered scenarios; in bold: method with the smallest absolute deviation from the nominal type I error level in the respective scenario}
	\tiny\centering
	\begin{tabular}{c c c c c c c c c c c}
		\toprule
		&&\multicolumn{9}{c}{$\delta=1.2$}\\
		\cmidrule(l{2pt}r{2pt}){3-11}
		&&\multicolumn{3}{c}{$m_0=1$}&\multicolumn{3}{c}{$m_0=2$}&\multicolumn{3}{c}{$m_0=4$}\\
		\cmidrule(l{2pt}r{2pt}){3-5}
		\cmidrule(l{2pt}r{2pt}){6-8}
		\cmidrule(l{2pt}r{2pt}){9-11}
		$\kappa$&Variance estimation&$n$&$\alpha$&$1-\beta$&$n$&$\alpha$&$1-\beta$&$n$&$\alpha$&$1-\beta$\\
		\midrule
		\multirow{4}{*}{0.1}&compensator&494&\textbf{.0240}&.8047&519&.0230&.8017&545&.0228&.8031\\
		&counting process&435&.0320&.7909&457&.0297&.7896&480&.0305&.7877\\
		&Wu&465&.0278&.7998&488&.0265&.7960&513&.0259&.7953\\
		&uncorrelated&475&.0268&.8013&500&\textbf{.0250}&.7985&526&\textbf{.0248}&.7987\\
		\hline
		\multirow{4}{*}{0.25}&compensator&454&.0226&.8071&510&.0231&.8063&578&.0226&.8053\\
		&counting process&400&.0308&.7925&449&.0311&.7901&509&.0296&.7893\\
		&Wu&427&.0261&.8010&480&.0268&.7987&543&.0262&.7960\\
		&uncorrelated&434&\textbf{.0251}&.8018&491&\textbf{.0253}&.8009&559&\textbf{.0242}&.7999\\
		\hline
		\multirow{4}{*}{0.5}&compensator&398&\textbf{.0239}&.8055&495&.0236&.8059&636&.0226&.8018\\
		&counting process&351&.0309&.7926&436&.0305&.7906&560&.0303&.7882\\
		&Wu&374&.0274&.7981&466&.0263&.7997&598&.0261&.7966\\
		&uncorrelated&377&.0270&.7993&475&\textbf{.0255}&.8014&617&\textbf{.0246}&.7993\\
		\hline
		\multirow{4}{*}{1}&compensator&325&.0226&.8105&466&.0230&.8076&767&.0230&.8026\\
		&counting process&287&.0302&.7943&410&.0300&.7919&675&.0310&.7887\\
		&Wu&306&\textbf{.0262}&.8016&438&.0262&.7993&721&.0270&.7953\\
		&uncorrelated&301&.0271&.7997&444&\textbf{.0255}&.8009&747&\textbf{.0246}&.7997\\
		\hline
		\multirow{4}{*}{2}&compensator&276&.0218&.8114&418&.0230&.8114&1065&.0233&.8053\\
		&counting process&244&.0289&.7959&369&.0300&.7970&937&.0306&.7888\\
		&Wu&260&\textbf{.0245}&.8030&393&\textbf{.0260}&.8038&1001&.0266&.7977\\
		&uncorrelated&248&.0279&.7977&392&.0261&.8036&1041&\textbf{.0242}&.8019\\
		\hline
		\multirow{4}{*}{5}&compensator&258&.0214&.8128&377&.0222&.8091&2057&.0231&.8023\\
		&counting process&228&.0282&.7962&333&.0290&.7953&1810&.0307&.7872\\
		&Wu&243&\textbf{.0250}&.8052&355&\textbf{.0253}&.8031&1934&.0266&.7931\\
		&uncorrelated&229&.0279&.7974&342&.0273&.7984&2016&\textbf{.0246}&.7985\\
		\midrule
		\multicolumn{2}{c}{}&\multicolumn{9}{c}{$\delta=1.5$}\\
		\cmidrule(l{2pt}r{2pt}){3-11}
		\multirow{4}{*}{0.1}&compensator&113&.0202&.8113&119&.0204&.8096&125&.0205&.8108\\
		&counting process&86&.0389&.7801&90&.0376&.7763&95&.0392&.7811\\
		&Wu&100&.0277&.7988&105&.0278&.7948&110&.0281&.7952\\
		&uncorrelated&104&\textbf{.0251}&.8017&110&\textbf{.0252}&.7998&117&\textbf{.0246}&.8045\\
		\hline
		\multirow{4}{*}{0.25}&compensator&104&.0195&.8173&117&.0200&.8110&133&.0206&.8101\\
		&counting process&78&.0364&.7823&88&.0369&.7769&100&.0384&.7761\\
		&Wu&91&.0271&.7998&103&.0274&.7974&117&.0285&.7956\\
		&uncorrelated&95&\textbf{.0245}&.8073&108&\textbf{.0244}&.8011&124&\textbf{.0247}&.8020\\
		\hline
		\multirow{4}{*}{0.5}&compensator&90&.0201&.8151&114&.0203&.8130&147&.0198&.8106\\
		&counting process&68&.0371&.7794&86&.0382&.7771&111&.0372&.7785\\
		&Wu&80&.0278&.8035&100&.0276&.7961&129&.0269&.7944\\
		&uncorrelated&81&\textbf{.0267}&.8045&104&\textbf{.0247}&.7996&138&\textbf{.0232}&.8020\\
		\hline
		\multirow{4}{*}{1}&compensator&73&.0191&.8230&106&.0211&.8125&178&.0208&.8070\\
		&counting process&56&.0359&.7914&81&.0384&.7829&134&.0386&.7754\\
		&Wu&64&\textbf{.0264}&.8054&94&.0282&.8021&156&.0285&.7932\\
		&uncorrelated&62&.0286&.8020&97&\textbf{.0265}&.8048&168&\textbf{.0239}&.8001\\
		\hline
		\multirow{4}{*}{2}&compensator&61&.0193&.8293&95&.0198&.8199&248&.0205&.8067\\
		&counting process&47&.0335&.7979&72&.0350&.7877&186&.0381&.7734\\
		&Wu&54&\textbf{.0254}&.8154&84&\textbf{.0264}&.8073&217&.0280&.7914\\
		&uncorrelated&49&.0312&.8027&83&.0270&.8049&236&\textbf{.0234}&.8007\\
		\hline
		\multirow{4}{*}{5}&compensator&56&.0183&.8283&84&.0195&.8288&480&.0209&.8087\\
		&counting process&44&.0337&.8021&65&.0353&.7979&361&.0385&.7781\\
		&Wu&50&\textbf{.0250}&.8168&74&\textbf{.0260}&.8128&421&.0285&.7942\\
		&uncorrelated&44&.0330&.7997&69&.0313&.8065&461&\textbf{.0231}&.8048\\
		\midrule
		\multicolumn{2}{c}{}&\multicolumn{9}{c}{$\delta=2$}\\
		\cmidrule(l{2pt}r{2pt}){3-11}
		\multirow{4}{*}{0.1}&compensator&46&.0185&.8230&48&.0180&.8172&51&.0169&.8204\\
		&counting process&28&.0502&.7624&30&.0484&.7657&31&.0502&.7597\\
		&Wu&37&.0302&.7939&39&.0297&.7936&41&.0294&.7922\\
		&uncorrelated&40&\textbf{.0250}&.8048&43&\textbf{.0244}&.8089&45&\textbf{.0240}&.8030\\
		\hline
		\multirow{4}{*}{0.25}&compensator&42&.0178&.9269&47&.0183&.8189&54&.0176&.8196\\
		&counting process&26&.0496&.7685&29&.0493&.7630&33&.0506&.7640\\
		&Wu&34&.0295&.8000&38&.0290&.7937&44&.0294&.7978\\
		&uncorrelated&36&\textbf{.0258}&.8057&42&\textbf{.0242}&.8106&48&\textbf{.0239}&.8040\\
		\hline
		\multirow{4}{*}{0.5}&compensator&36&.0180&.8251&46&.0178&.8220&60&.0183&.8192\\
		&counting process&23&.0463&.7768&28&.0491&.7607&37&.0500&.7628\\
		&Wu&29&.0289&.7964&37&.0293&.7941&48&.0299&.7906\\
		&uncorrelated&30&\textbf{.0272}&.8039&40&\textbf{.0249}&.8032&54&\textbf{.0234}&.8047\\
		\hline
		\multirow{4}{*}{1}&compensator&29&.0170&.8393&43&.0179&.8256&72&.0181&.8131\\
		&counting process&18&.0447&.7765&27&.0470&.7728&44&.0516&.7550\\
		&Wu&24&\textbf{.0272}&.8190&35&.0287&.8013&59&.0310&.7950\\
		&uncorrelated&22&.0310&.8012&37&\textbf{.0260}&.8091&66&\textbf{.0226}&.8029\\
		\hline
		\multirow{4}{*}{2}&compensator&23&.0158&.8441&38&.0169&.8353&101&.0181&.8111\\
		&counting process&15&.0410&.7907&24&.0460&.7801&62&.0513&.7555\\
		&Wu&19&\textbf{.0257}&.8189&31&\textbf{.0281}&.8105&82&.0302&.7894\\
		&uncorrelated&17&.0356&.8167&31&.0289&.8143&94&\textbf{.0218}&.8037\\
		\hline
		\multirow{4}{*}{5}&compensator&21&.0157&.8509&33&.0160&.8536&198&.0181&.8152\\
		&counting process&14&.0387&.7988&22&.0401&.8035&121&.0515&.7589\\
		&Wu&18&\textbf{.0244}&.8348&28&\textbf{.0255}&.8361&161&.0311&.7937\\
		&uncorrelated&15&.0373&.8164&24&.0333&.8140&186&\textbf{.0219}&.8097\\
		\bottomrule
	\end{tabular}
	\label{table:samplesizes_emp_erorrs}
\end{table}
As already seen in the previous subsection, our approach works very well for small and medium event rates (values for $\kappa$ and $m_0$ resulting in event rates $\leq 70\%$ according to Table \ref{table:expected_events_and_weights}) and is the best choice concerning the absolute deviance from the nominal type I error level in most of the cases. The first exception concerns scenarios with high event rates (values for $\kappa$ and $m_0$ resulting in event rates $> 70\%$ according to Table \ref{table:expected_events_and_weights}) where the approach by Wu \cite{Wu:2014} is preferred. Only in cases with larger sample sizes the original version of the one-sample log-rank test performs best in means of absolute deviation from the nominal type I error.\\
As one can already see from the flexible sample size formula \eqref{eq:sample_size}, the approach using only the counting process to estimate the variance requires the smallest sample sizes. Nevertheless the type I error is inflated in any scenario, ranging from 2.82\% to 5.16\%. The original variance estimation just behaves the other way round. Obviously, the required sample size is always the highest, while the type I error rate tends to be deflated. It ranges between 1.57\% and 2.4\%, depending on the scenario. The remaining approaches are located in between concerning the sample size whereby our proposed approach requires higher sample sizes if and only if the weight (see Table \ref{table:expected_events_and_weights}) is smaller than $1/2$.\\
Concerning the compliance with the given type II error rate, the uncorrelated variance estimation works best on average with empirical power ranging between 79.7\% and 81.7\%. Here, also the highest deviations occur for high event rates, i.e. in scenarios in which Wu's approach is also superior concerning the empirical type I error. These results confirm that a combination of Wu's and our uncorrelated variance estimator, as given in \eqref{eq:combination_Wu_uncorrelated}, promises the most satisfying performance.
\section{Practical example}\label{section:real_world_example}
We illustrate the differences of the approaches using a practical example. We employ the setting of the Mayo Clinic trial in primary biliary cirrhosis of the liver (PBC), which is a rare but fatal chronic disease whose cause is still unknown \cite{Fleming:2011}. In this double-blinded randomized trial the drug D-penicillamine (DPCA) was compared with a placebo. The study data is publicly available via the survival package in R \cite{R, Therneau:2020}.\\
Among the 158 patients of the cohort treated with DPCA, 65 died during the trial. For the sake of comparability, we adopt the already established parameter estimation of their survival curve \cite{Wu:2015}, which states that a Weibull distribution with shape parameter $\kappa=1.22$ and median survival time $m_0=9$ years fits the data well. We now suppose, that a new treatment becomes available and the data from this trial shall be used to compare the survival under this new treatment to the one under DPCA. This shall be accomplished in a trial in which patients are recruited uniformly over a accrual period of length $a=5$ years and followed-up in an additional period of length $f=3$ years. As in the preceding simulations, the planning hypothesis is supposed to fulfill the proportional hazards assumption and hence $S_{T,1}(s)^{\delta}=S_{T,0}(s)$ for any $s\geq 0$. A hazard ratio of $\delta=1.75$ shall be detected with a power of $1-\beta=0.8$ via a two-sided test with significance level $2\alpha=0.05.$\\
\begin{table}
	\caption{Sample sizes and empirical type I \& type II errors for the planning of a single-arm survival trial for PBC, where $\alpha$ denotes the two-sided empirical type I error and $\alpha_l$ denotes the empirical type I error for the left-sided testing problem $H_{0,1}$.}
	\centering
	\begin{tabular}{l c c c c}
		\toprule
		Variance estimation & $n$ & $\alpha$ & $\alpha_l$ & $1-\beta$\\
		\midrule
		compensator&113&.0504&.0193&.8120\\
		counting process&76&.0578&.0455&.7671\\
		Wu&95&.0511&.0293&.7931\\
		uncorrelated&106&.0493&.0225&.8045\\
		\bottomrule
	\end{tabular}
	\label{table:real_world_example}
\end{table}
The equation \eqref{eq:ordinary_weight} yields a prefixed weight of $w_0(8)=0.1923$ for this scenario and hence our suggested test statistic amounts to
\begin{equation}
\frac{N(8)-A_0(8)}{\sqrt{0.8077 \cdot A_0(8) + 0.1923 \cdot N(8)}}.
\end{equation}
The results of the following simulation with 100 000 runs, which are displayed in Table \ref{table:real_world_example} show that in this real world example, our proposed approach is very close to the nominal type I error level in terms of empirical type I error of the two-sided test as well as the left-sided test. The original one-sample log-rank test and Wu's suggestion look similar for the two-sided testing problem (column $\alpha$ in Table \ref{table:real_world_example}) while they perform remarkably worse in case of the left-sided test (column $\alpha_l$ in Table \ref{table:real_world_example}). The phenomenon of unbalanced left- and right-sided type I errors which cancel each other out quite well in their sum is remarkable here.\\
Although the sample size for our suggested approach is about 10\% higher than for Wu's suggestion, there is still a saving in sample size compared to the standard approach.
\section{Discussion}\label{section:discussion}
We introduced a simple but extensive framework, enabling a continuum of consistent variance estimators referring to the one-sample log-rank test. Asymptotical correctness as well as power and sample size calculations are provided. The classical one-sample log-rank test \cite{Breslow:1975} and a practical alternative \cite{Wu:2014} naturally fit into the framework. 
In addition, we elaborated special choices for the weight function $w$ which guarantees that the variance estimator is uncorrelated to the compensated counting process under the null hypothesis or under a planning alternative, respectively. In several simulations and in an example based on real world data, we can see that the emerging test is superior to other approaches concerning the adherence to the nominal type I error level. This superiority is most remarkable in small sized trials with small to medium event rates. 
It is important to note that our approach formally depends on planning assumptions about accrual and dropout mechanism, but that misspecifications of those do not impair the asymptotical correctness of the resulting one-sample log-rank test. The effects of such misspecifications on the weight have been demonstrated in Section \ref{section:misspecifications}. However, both the theoretical results and the simulations in Section \ref{section:simulations_variable_weight} show that they do not affect the asymptotical properties of the test in the large sample limit.\\ 
Nevertheless, we saw in our simulation studies that Wu's suggested weight works better than the uncorrelated variance estimation in scenarios with high event rates. All in all, this suggests to choose the weight according to \eqref{eq:combination_Wu_uncorrelated}.\\
One might also pursue a simulation-based approach to find the perfect weight for the envisaged scenario which can cancel out a possible skewness of $M_0(t)$ under the null hypothesis. In this case, the theory from Section \ref{section:definitions} provides the asymptotical correctness and sample size calculation can be done as lined out in Section \ref{section:sample_size}.\\ 
From a mathematical point of view, it would even be allowed to choose a random weight $W$ as long as it can be guaranteed that $W\in [0,1]$. This is proven in Lemma \ref{lemma:random_weight} of the Appendix. It can be of use for non-clinical trial data or in other scenarios in which we do not have any prior information available for the accrual and censoring distributions or one does not want to use a prespecified weight because of uncertainty about these mechanisms. Then we can use information from the trial itself to compute
\begin{equation}
W(t)\coloneqq 1-\frac{\int_0^t S_{T,0}(s)\Lambda_{T,0}(s)d\hat{F}_{U}(s)}{\int_0^t F_{T,0}(s)d\hat{F}_{U}(s)}
\end{equation}
where $\hat{F}_{U}$ is the Kaplan-Meier estimator of the distribution function of the random variable $U \coloneqq C \wedge (t-Y)_+$. With the well-known properties of such Kaplan-Meier integrals \cite{Stute:1995}, we can state the convergence
\begin{equation}
W(t) \underset{n\to\infty}{\to} 1-\frac{\int_0^t S_{T,0}(s)\Lambda_{T,0}(s)dF_{U}(s)}{\int_0^t F_{T,0}(s)dF_{U}(s)} \quad \mathbb{P}\text{-almost surely.}
\end{equation}
In any case, the way in which the weight will be determined, should be specified in advance and must not be changed while the trial is ongoing.\\
In conclusion, our framework yields a solid foundation for such and possible further considerations. This includes extensions to multi-stage \cite{Shan:2019, Belin:2017, Kwak:2014}, multivariate \cite{Danzer:2021} and other variations \cite{Chu:2020} of the classical one-sample log-rank test.

\section*{Appendix}
\subsection*{Non-positivity of Covariance of $A_0(t)$ and $N(t)$}
As one can see from \eqref{eq:defining_equation} we need to show that $\text{Cov}_{H_j}(A_0(t), N(t))\leq 0$ if we want to state that $w_j(t) \in [0,1]$. First of all, we have
\begin{equation*}
\text{Cov}_{H_j}(A_0(t), N(t)) = \sum_{i=1}^n \text{Cov}_{H_j}(A_{0,i}(t), N_i(t))
\end{equation*}
by the independence of our observations. As the observations are also identically distributed it suffices to show $\text{Cov}_{H_j}(A_{0,1}(t), N_1(t))\leq0$ which will be accomplished by the following Lemma. 
\begin{lemma}\label{lemma:negative_covariance}
Let $(A_{0,1}(t))_{t \in \mathbb{R}_+}$ and $(N_1(t))_{t \in \mathbb{R}_+}$ be the two real-valued stochastic processes as defined in \eqref{eq:counting_process} resp. \eqref{eq:compensator_under_H0} and let a fixed hypothesis $H_j$ about the distribution of the $T_i$ be given. Then we have
\begin{equation*}
\text{Cov}_{H_j}(A_{0,1}(t), N_1(t))\leq 0
\end{equation*}
for any $t\in\mathbb{R_+}$ and $j\in\{0,1\}$.
\end{lemma}

\begin{proof}
\begin{align*}
&\text{Cov}_{H_j}(A_{0,1}(t), N_1(t))\\
=&\mathbb{E}[A_{0,1}(t)\cdot N_1(t)] - \mathbb{E}_{H_j}[A_{0,1}(t)]\cdot \mathbb{E}_{H_j}[N_1(t)]\\
=&\mathbb{E}[A_{0,1}(t); N_1(t)=1] - \mathbb{E}_{H_j}[A_{0,1}(t)]\cdot \mathbb{E}_{H_j}[N_1(t)]\\
=&\mathbb{E}[A_{0,1}(t)|N_1(t)=1]\cdot \mathbb{E}_{H_j}[N_1(t)] - \mathbb{E}_{H_j}[A_{0,1}(t)]\cdot \mathbb{E}_{H_j}[N_1(t)]\\
=&(\mathbb{E}[A_{0,1}(t)|N_1(t)=1] - \mathbb{E}_{H_j}[A_{0,1}(t)])\cdot \mathbb{E}_{H_j}[N_1(t)]
\end{align*}
as $N_1(t) \in \{0,1\}$ almost surely. The second factor is non-negative and hence we need to check the negativity of the first factor. As
\begin{align*}
\mathbb{E}_{H_j}[A_{0,1}(t)|N_1(t)=1]&=\mathbb{E}_{H_j}[\Lambda_{0,T}(T_1\wedge C_1 \wedge (t-Y_1)_+) | T_1 \leq C_1 \wedge (t-Y_1)_+]\\
&=\mathbb{E}_{H_j}[\Lambda_{0,T}(T_1) | T_1 \leq C_1 \wedge (t-Y_1)_+]
\end{align*}
and 
\begin{equation*}
\mathbb{E}_{H_j}[A_{0,1}(t)]=\mathbb{E}_{H_j}[\Lambda_{0,T}(T_1\wedge C_1 \wedge (t-Y_1)_+)]
\end{equation*}
and $\Lambda_{0,T}$ is a non-decreasing function, $\mathbb{E}_{H_j}[A_{0,1}(t)|N_1(t)=1] \leq \mathbb{E}_{H_j}[A_{0,1}(t)]$ holds if $T_1|T_1 \leq C_1 \wedge (t-Y_1)_+ \preceq T_1 \wedge C_1 \wedge (t-Y_1)_+$. For any $s\geq 0$ we have by the law of total probability
\begin{align*}
&\mathbb{P}_{H_j}[T_1\leq s|T_1\leq C_1\wedge (t-Y_1)_+]\\
=&\mathbb{P}_{H_j}[T_1\leq s|T_1\leq C_1\wedge (t-Y_1)_+, C_1\wedge(t-Y_1)_+ \leq s] \cdot \mathbb{P}_{H_j}[C_1 \wedge (t-Y_1)_+ \leq s]\\
&+\mathbb{P}_{H_j}[T_1\leq s|T_1\leq C_1\wedge (t-Y_1)_+, C_1\wedge(t-Y_1)_+ > s] \cdot \mathbb{P}_{H_j}[C_1 \wedge (t-Y_1)_+ > s]\\
=&\mathbb{P}_{H_j}[C_1 \wedge (t-Y_1)_+ \leq s]\\
&+\mathbb{P}_{H_j}[T_1\leq s|T_1\leq C_1\wedge (t-Y_1)_+, C_1\wedge(t-Y_1)_+ > s] \cdot \mathbb{P}_{H_j}[C_1 \wedge (t-Y_1)_+ > s]
\end{align*}
resp.
\begin{align*}
&\mathbb{P}_{H_j}[T_1 \wedge C_1\wedge (t-Y_1)_+ \leq s]\\
=&\mathbb{P}_{H_j}[T_1 \wedge C_1\wedge (t-Y_1)_+ \leq s|C_1\wedge(t-Y_1)_+ \leq s] \cdot \mathbb{P}_{H_j}[C_1 \wedge (t-Y_1)_+ \leq s]\\
&+\mathbb{P}_{H_j}[T_1 \wedge C_1\wedge (t-Y_1)_+ \leq s|C_1\wedge(t-Y_1)_+ > s] \cdot \mathbb{P}_{H_j}[C_1 \wedge (t-Y_1)_+ > s]\\
=&\mathbb{P}_{H_j}[C_1 \wedge (t-Y_1)_+ \leq s]\\
&+\mathbb{P}_{H_j}[T_1\leq s] \cdot \mathbb{P}_{H_j}[C_1 \wedge (t-Y_1)_+ > s].
\end{align*}
Now, we can conclude with Bayes' theorem:
\begin{align*}
&\mathbb{P}_{H_j}[T_1\leq s|T_1\leq C_1\wedge (t-Y_1)_+, C_1\wedge(t-Y_1)_+ > s]\\
=& \mathbb{P}_{H_j}[T_1 \leq C_1\wedge (t-Y_1)_+|T_1\leq s, C_1\wedge(t-Y_1)_+ > s]\\
&\cdot \frac{ \mathbb{P}_{H_j}[T_1\leq s| C_1\wedge(t-Y_1)_+ > s]}{\mathbb{P}_{H_j}[T_1 \leq C_1\wedge (t-Y_1)_+| C_1\wedge (t-Y_1)_+ > s]}\\
=&\frac{\mathbb{P}_{H_j}[T_1\leq s]}{\mathbb{P}_{H_j}[T_1 \leq C_1\wedge (t-Y_1)_+| C_1\wedge (t-Y_1)_+ > s]}\\
\geq & \mathbb{P}_{H_j}[T_1\leq s].\\
\end{align*}
\end{proof}

\subsection*{Consistency of variance estimation with random weights}
As proposed in Section \ref{section:discussion}, we can also choose the weight randomly at the analysis stage based on the observed censoring and accrual mechanism. The following Lemma guarantees that such a procedure also leads to a consistent estimation of the variance and hence to an asymptotically correct testing procedure.

\begin{lemma}\label{lemma:random_weight}
Let $(X_n)_{n\in\mathbb{N}}$ and $(Y_n)_{n\in\mathbb{N}}$ be two sequences of real-valued random variables with $X_n \overset{\mathbb{P}}{\to}c$ and $Y_n \overset{\mathbb{P}}{\to}c$ as $n\to\infty$ for some constant $c \in \mathbb{R}$. Let $(W_n)_{n\in\mathbb{N}}$ be another sequence of random variables. The support of any of these random variables is a subset of $[0,1]$, i.e. $0 \leq W_n \leq 1$ a.s. for any $n \in \mathbb{N}$. Then
\begin{equation}
W_n\cdot X_n + (1-W_n)\cdot Y_n \overset{\mathbb{P}}{\to} c
\end{equation}
as $n\to\infty$.
\end{lemma}

\begin{proof}
Let $\varepsilon > 0$ arbitrary. From the triangle inequality, we get
\begin{align}
&\mathbb{P}[|W_n(X_n-c)+(1-W_n)(Y_n-c)|>\varepsilon]\\
\leq & \mathbb{P}[|W_n(X_n-c)|+|(1-W_n)(Y_n-c)|>\varepsilon]
\end{align}
Now $|W_n(X_n-c)|+|(1-W)(Y_n-c)|>\varepsilon$ only holds if at least one of the terms $|W_n(X_n-c)|$ and $|(1-W)(Y_n-c)|$ exceeds $\varepsilon/2$. With subadditivity and the properties of the $W_n$, we can conclude with
\begin{align}
& \mathbb{P}[|W_n(X_n-c)|+|(1-W_n)(Y_n-c)|>\varepsilon]\\
\leq &\mathbb{P}[|W_n(X_n-c)|>\varepsilon/2]+\mathbb{P}[|(1-W_n)(Y_n-c)|>\varepsilon/2]\\
\leq &\mathbb{P}[|(X_n-c)|>\varepsilon/2]+\mathbb{P}[|(Y_n-c)|>\varepsilon/2]
\end{align}
as both of these summands converge to zero by our assumptions.
\end{proof}

\section*{Acknowledgments}
The work of the corresponding author was supported by the German Science Foundation (DFG,
grant number 413730122).

\bibliographystyle{unsrt}

\end{document}